\documentclass[journal]{IEEEtran}

\ifCLASSINFOpdf
\else
   \usepackage[dvips]{graphicx}
\fi
\usepackage{url}

\usepackage{graphicx}
\usepackage{subcaption}

\usepackage{amsthm,amsmath,amsfonts,amssymb,breqn}
\usepackage[numbers, square, comma]{natbib}
\usepackage[dvipsnames]{xcolor}
\RequirePackage[colorlinks,citecolor=RoyalBlue]{hyperref}

\usepackage{caption}
\usepackage{comment}

\theoremstyle{plain}
\newtheorem{theorem}{Theorem}[section]
\newtheorem{lemma}{Lemma}[section]
\newtheorem{proposition}{Proposition}[section]
\newtheorem{corollary}{Corollary}[section]

\theoremstyle{remark}

\newtheorem{remark}{Remark}

\newcommand{\Pro}{\mathrm{P}}
\newcommand{\Qro}{\mathrm{Q}}
\newcommand{\Exp}{\mathrm{E}}

\newcommand{\cL}{\mathcal{L}}
\newcommand{\cF}{\mathcal{F}}

\newcommand{\bb}{\boldsymbol{b}}
\DeclareMathOperator*{\argmax}{arg\,max\;}

\newcommand{\btheta}{\boldsymbol{\theta}}

\newcommand{\bL}{\boldsymbol{L}}
\newcommand{\cS}{\mathcal{S}}
\newcommand{\Alt}{\operatorname{Alt}}
\newcommand{\Err}{\operatorname{Err}}

\begin{document}

\title{Signal Detection under Composite Hypotheses with Identical Distributions for Signals and for Noises}

\author{Yiming Xing,~\IEEEmembership{Member, IEEE}, 
Anamitra Chaudhuri,~\IEEEmembership{Member, IEEE}, 
and Yifan Chen,~\IEEEmembership{Member,~IEEE}
\thanks{This paragraph of the first footnote will contain the date on which you submitted your paper for review.}
\thanks{Due to page limit, the Supplementary Material is posted at \cite{SPL_supplement}.}
\thanks{
Yiming Xing is with the School of Mathematical Sciences, Tongji University, Shanghai, China (email: yimingx4@tongji.edu.cn).
}
\thanks{Anamitra Chaudhuri is with the Department of Statistics, Texas A\&M University, Texas, USA (e-mail: ac27@tamu.edu).}
\thanks{Yifan Chen is with the Department of Computer Science, Hong Kong Baptist University, Hong Kong, China (e-mail: yifanc@hkbu.edu.hk).
}
}

\markboth{Journal of \LaTeX\ Class Files, Vol. 14, No. 8, August 2015}
{Shell \MakeLowercase{\textit{et al.}}: Bare Demo of IEEEtran.cls for IEEE Journals}
\maketitle

\begin{abstract}
In this paper, we consider the problem of detecting signals in multiple, sequentially observed data streams, where the distribution of each stream lies in one of two common composite spaces, depending on whether it is a signal or a noise.
For this problem, we study a practical yet underexplored setting where it is a priori known that all signals have an identical distribution and so do all noises. 
Compared to the general setting where local distributions are free to take any values, this structure facilitates faster decision-making thanks to a smaller joint distribution space. 
However, it introduces additional challenges to the analysis of problem and design of tests, since the local distributions are now coupled. 
{\color{blue}
In this paper, we first establish a universal lower bound on the minimum expected sample size, which characterizes the essential difficulty of the problem and involves constants that are neither the minimum Kullback-Leibler divergences between the signal/noise distribution to the noise/signal distribution space, which appear in the lower bound for the general setting, nor the Kullback-Leibler divergences between the signal distribution and the noise distribution.}
Besides, we propose a test that controls the two types of familywise error rates below arbitrary levels, and achieves the minimum expected sample size asymptotically as the levels go to zero. 
{\color{blue}
Numerical studies are presented to compare with the state-of-the-art test for the general setting and demonstrate robustness against model misspecification.
}


\end{abstract}

\begin{IEEEkeywords}
    Asymptotic optimality, composite hypotheses, sequential multiple testing, signal detection, structured hypotheses.
\end{IEEEkeywords}

\IEEEpeerreviewmaketitle

\vspace{-1ex}
\section{Introduction}
Detecting signals over multiple data streams that are observed sequentially in real-time is a fundamental problem in signal processing and its related fields.
For example, in air defense systems, we aim to detect missile intrusions in multiple areas \citep{Ref4MissleDet_2004}; 
in the development of precision medicines, we aim to identify effective targets across various positions
\citep{Bartroff_Book_Clinicaltrials}; in spectrum sensing for cognitive radio, we aim to find vacant channels \citep{spectrumsensing_2012}; in financial markets, we aim to monitor frauds \cite{Ref4FraudDet_2022}, etc.
If the characteristics of signals are specified as alternative hypotheses and those of non-signals, i.e., noises, as null hypotheses, then such a problem can be naturally formulated as a sequential multiple testing problem.

Such a problem has been studied in \cite{Malloy_Nowak_2014, Kobi_2015_active, Kobi_2015_AO, Kobi_2018_heterogeneous, Kobi_2019_nonlinear, Kobi_2022_switching} and \cite{Song_prior, PaperII, Aris_IEEE, Aris_TIT2025, ITW2024} where the hypotheses are \emph{simple}, i.e., the distributions of each stream under the null and the alternative hypotheses 
are fully specified, and in \cite{Kobi_2020_composite, Kobi_2023_hiera} and \cite{Song_AoS, PaperIII, chaudhuri2024joint, chaudhuriISIT2024} where the hypotheses are \emph{composite}, i.e., the distributions of each stream are only specified up to an unknown parameter.

{\color{blue}
In most works considering composite hypotheses, the local parameter in each stream is allowed to take any values within its parameter space.
However, there are numerous scenarios where all signals share a common parameter and all noises share another. 
A typical example is the post-change identification problem, where all streams are initially in a common normal state, an abrupt event changes a subset of them to a common abnormal state, and it is of interest to identify this subset reliably and efficiently.
Based on the authors' knowledge, this very realistic setting has been considered only in \citep[Section III-C]{Kobi_2020_composite}, where a test was proposed and was only shown to be consistent in the sense that its familywise misclassification rate decays polynomially with its threshold (in the scale of log-likelihood ratio). 

Our contributions of this work are the rigorous formulation of this setting, and the proposal of a novel test tailored for it.
For the proposed test, we show that
(i) its two types of familywise error rates decay exponentially with its thresholds, and we design a universal selection of the thresholds so that the error rates can be controlled below arbitrary, user-specified levels, and 
(ii) it is asymptotically optimal, in the sense that its expected sample size achieves the infimum among all tests that control the error rates below the same levels, asymptotically as the levels go to zero.

Numerical studies are presented in both the main text and the supplement to illustrate the properties of the proposed test. An extension to various other error metrics, all proofs and some supporting lemmas are also presented in the supplement.
}

\vspace{-1ex}
\section{Problem formulation} \label{section: problem formulation}
Let $\{X_k(n),\,n\geq 1\}$, $k\in[K]\equiv\{1,\ldots,K\}$ be $K\geq 1$ independent data streams, each comprising i.i.d. data.
Suppose that the local distributions of all streams 
belong to the same parameter family, with densities belonging to  $\{f_\theta,\,\theta\in\Theta\}$ with respect to certain $\sigma$-finite measure $\nu$.
For any $k\in[K]$, we call stream $k$ as a \textit{noise} (resp. \textit{signal}) if the corresponding local parameter $\theta_k \in \Theta^0$ (resp. \textit{$\Theta^1$}),
where $\Theta^0$ and $\Theta^1$ are disjoint non-empty subsets forming a partition of $\Theta$.
We denote by $\btheta = (\theta_1, \dots, \theta_K)$ the joint parameter, 
by $\Theta^K$ the joint parameter space and by $\cS$ the structured subset of the joint parameter space where the local parameters of all noises and of all signals are identical, respectively, i.e., 
\begin{equation*}
\begin{gathered}
    \cS \equiv \{\btheta\in\Theta^K: \;
    \exists\;A\subseteq[K],\; \theta^0\in\Theta^0,\;\theta^1\in\Theta^1, \\
    \text{such that } \theta_k=\theta^0 \text{ for } k\in[K]\backslash A \text{ and } \theta_k=\theta^1 \text{ for } k\in A\}.
\end{gathered}
\end{equation*}


The problem of interest is, based on data that are sampled sequentially in time, to identify the subset of signals with desired reliability and as quickly as possible.
In order to do this, we need to specify a random time $T$ and a random set $D\subseteq[K]$ so that after taking $T$ samples in each stream, we stop sampling and declare $D$ as the subset of signals.
It is natural to require that when to stop sampling and which streams to select as signals should be based only on the already sampled data.
Mathematically, by denoting the data filtration by $\cF\equiv\{\cF(n),\,n\geq 1\}$ where $\cF(n) \equiv \sigma(X_k(t),\,1\leq t\leq n,\,k\in[K])$, this is achieved by requiring $T$ be a stopping time with respect to $\cF$ and $D$ be $\cF(T)$-measurable, i.e., for any $n\geq 1$ and $B\subseteq[K]$, $\{T\leq n\}, \, \{T\leq n,\,D=B\}\in\cF(n)$.
We refer to such a tuple $(T,D)$ as a test and denote by $\Delta$ the family of all tests.

For any $\btheta\in\Theta^K$, denote by $A(\btheta)$ the subset of streams with signals, i.e., $A(\btheta) \equiv \{k\in[K]: \theta_k\in\Theta^1\}$.
Therefore, for any $(T,D)\in\Delta$, $D\backslash A(\btheta)$ represents the subset of streams that are noises but are misidentified as signals, i.e., where type-I errors are made, and $A(\btheta)\backslash D$ represents the subset of streams that are signals but are misidentified as noises, i.e., where type-II errors are made. 
We are interested in controlling the probabilities of both types of errors.
Specifically, for any $\alpha,\beta\in(0,1)$, we denote by $\Delta(\alpha,\beta)$ the subfamily of tests that terminate almost surely and control the two types of familywise error rates below $\alpha,\beta$, respectively, under every possible distribution, i.e.,
\begin{equation*}
\begin{gathered}
    \Delta(\alpha,\beta) \equiv \{
    (T,D)\in\Delta: \;
    \Pro_\btheta(T<\infty) = 1, \\
    \Pro_\btheta(D\backslash A(\btheta)\neq\emptyset) \leq \alpha, 
    \Pro_\btheta(A(\btheta)\backslash D\neq\emptyset) \leq \beta, \; \forall\; \btheta\in\cS
    \}.
\end{gathered}
\end{equation*}
Under these constraints, our goal is to minimize the expected sample size under every possible distribution, i.e., to achieve
\begin{equation*}
    \cL_\btheta(\alpha,\beta) \equiv \inf\{\Exp_\btheta[T]: (T,D)\in\Delta(\alpha,\beta)\}, \;\forall\;\btheta\in\cS,
\end{equation*}
to a first-order asymptotic approximation as $\alpha,\beta\to 0$.

\begin{remark}
To facilitate later comparisons, we also introduce analogous notations for the general, unstructured setting.
Specifically, for any $\alpha,\beta\in(0,1)$, we denote by $\tilde{\Delta}(\alpha,\beta)$ the same as $\Delta(\alpha,\beta)$ with the only difference that the conditions hold for all $\btheta\in\Theta^K$, and, for any $\btheta\in\Theta^K$, denote by $\tilde\cL_\btheta(\alpha,\beta)$ the same as $\cL_\btheta(\alpha,\beta)$ with the only difference that the infimum is taken with respect to all $(T,D)\in\tilde\Delta(\alpha,\beta)$.
\end{remark}

\section{Universal lower bound} \label{section: ALB}
\subsection{Notations and assumptions}
For any $\theta,\theta'\in\Theta$, denote by 
\begin{equation*}
    I(\theta,\theta') \equiv \int f_{\theta} \log\frac{f_{\theta}}{f_{\theta'}} d\nu
\end{equation*}
the Kullback-Leibler (KL) divergence between $\theta$ and $\theta'$, i.e., between $f_\theta$ and $f_{\theta'}$.
Assume that the two parameter spaces, $\Theta^0$ and $\Theta^1$, are separated, in the sense that 
\begin{equation*}
\begin{aligned}
    I(\theta^0,\Theta^1) & \equiv \inf_{\theta\in\Theta^1} I(\theta^0,\theta) > 0, && \forall\; \theta^0\in\Theta^0, \\
    I(\theta^1,\Theta^0) & \equiv \inf_{\theta\in\Theta^0} I(\theta^1,\theta) > 0, && \forall\; \theta^1\in\Theta^1.
\end{aligned}
\end{equation*}
Besides, for any joint parameters $\btheta,\btheta'\in\Theta^K$, denote by $I(\btheta,\btheta')$ the KL divergence  between $\btheta$ and $\btheta'$, i.e., between the joint distributions $f_\btheta=\prod_{k\in[K]}f_{\theta_k}$ and $f_{\btheta'}=\prod_{k\in[K]}f_{\theta'_k}$. 
Due to independence across streams, we have 
$I(\btheta,\btheta') = \sum_{k\in[K]} I(\theta_k,\theta'_k)$.


\subsection{Universal lower bound}
{\color{blue} Define the function 
$$\varphi(x,y) \equiv x\log(x/(1-y)) + (1-x)\log((1-x)/y)$$
for $x,y\in(0,1)$ that $x+y<1$, which is decreasing in both arguments and $\sim|\log y|$ as $x,y\to 0$. 
This function is common in the universal lower bounds for sequential problems, due to the application of an information-theoretical inequality.
See \citep[Chapter 3.2.1]{Tartakovsky_Book}, \citep[Section 2]{OAI} and \citep[Section 5]{Song_prior} for more examples.}

Besides, 
for any $A\subseteq[K]$,
we denote by
\begin{equation*}
\begin{aligned}
    \Alt^0(A) & \equiv \{\btheta'\in\cS: A(\btheta')\backslash A\neq\emptyset \}, \\
    \Alt^1(A) & \equiv \{\btheta'\in\cS: A\backslash A(\btheta')\neq\emptyset \},
\end{aligned}
\end{equation*}
the subset of the structured parameter space which makes type-I (resp. II) errors relative to $A$,
{\color{blue}
or equivalently, relative to which $A$ makes type-II (resp. I) errors,
where $\Alt$ stands for ``alternative".
Note that, although the former interpretation is more straightforward, the latter is more precise because the principle of this problem (and probably of all hypothesis testing problems) is that, we select $A$ after making sure that the risk of making type-II (resp. I) errors if the truth is in $\Alt^0(A)$ (resp. $\Alt^1(A)$) is low.}
Finally, for any $\btheta\in\cS$, we denote by
\begin{equation*}
\begin{aligned}
    I^i(\btheta) & \equiv \inf_{\btheta'\in\Alt^i(A(\btheta))} I(\btheta,\btheta') \text{ for } i\in\{0,1\}, 
\end{aligned}
\end{equation*}
the minimum distances between $\btheta$ and $\Alt^i(A(\btheta))$, measured by KL divergence, with the convention that the infimum over an empty set of non-negative numbers is $+\infty$. 
{\color{blue}
\begin{remark}
    Similarly, we denote $\tilde\Alt^i(A)$ as $\Alt^i(A)$ with $\btheta'\in\Theta^K$, and $\tilde I^i(\btheta)$ as $I^i(\btheta)$ with $\btheta'\in\tilde\Alt^i(A(\btheta))$. 
\end{remark}}

\begin{theorem} \label{theorem, LB}
    For every $\btheta\in\cS$ and $\alpha,\beta\in(0,1/2)$ that $\alpha+\beta<1/2$,
    \begin{equation*} \label{LB}
    \begin{aligned}
        \cL_\btheta(\alpha,\beta) \geq \max \Bigg\{ \frac{\varphi(\alpha+\beta,\beta)}{I^0(\btheta)}, \frac{\varphi(\alpha+\beta,\alpha)}{I^1(\btheta)} \Bigg\}.
    \end{aligned}   
    \end{equation*}
    Therefore, as $\alpha,\beta\to 0$,
    \begin{equation*} \label{asy LB}
        \cL_\btheta(\alpha,\beta) \gtrsim \max\left\{ \frac{|\log\beta|}{I^0(\btheta)}, \frac{|\log\alpha|}{I^1(\btheta)} \right\}.
    \end{equation*}
\end{theorem}

    From the literature of sequential multiple testing with no structure on the joint parameter, e.g., \cite{Song_prior}, we know that, for any $\btheta\in\Theta^K$, as $\alpha,\beta\to 0$,
    the optimal performance
    \begin{equation} \label{tilde Lbtheta}
        \tilde\cL_\btheta(\alpha,\beta)\sim\max \left\{ \frac{|\log\beta|}{\tilde I^0(\btheta)}, \frac{|\log\alpha|}{\tilde I^1(\btheta)} \right\},
    \end{equation}
    where
    \begin{equation*}
        \tilde I^0(\btheta) = \min_{k\in[K]\backslash A(\btheta)} I(\theta_k,\Theta^1), \;\;
        \tilde I^1(\btheta) = \min_{k\in A(\btheta)} I(\theta_k,\Theta^0).
    \end{equation*}
    By definition, it is clear that $I^0(\btheta) \geq \tilde I^0(\btheta)$ and $I^1(\btheta) \geq \tilde I^1(\btheta)$ for all $\btheta\in\cS$.

{\color{blue}
\subsection{Comparison}
In this subsection, we draw sketches and simplify the expressions in the Gaussian case to provide more intuitions about the comparison between $I^1(\btheta)$ and $\tilde I^1(\btheta)$.
In Fig.~\ref{Fig: sketches} we draw four joint parameters, whose meanings are explained in the caption.
\begin{figure}[htbp]
    \centering
    
    \begin{subfigure}[b]{0.49\linewidth}
        \centering
        \includegraphics[width=\linewidth]{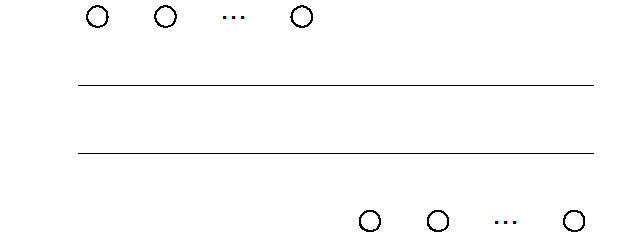}
        \subcaption{}
        \label{subfig:1}
    \end{subfigure}
    \hfill
    \begin{subfigure}[b]{0.49\linewidth}
        \centering
        \includegraphics[width=\linewidth]{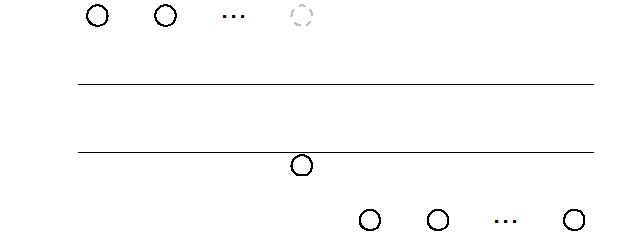}
        \subcaption{}
        \label{subfig:2}
    \end{subfigure}
    
    \vspace{5mm}
    
    \begin{subfigure}[b]{0.49\linewidth}
        \centering
        \includegraphics[width=\linewidth]{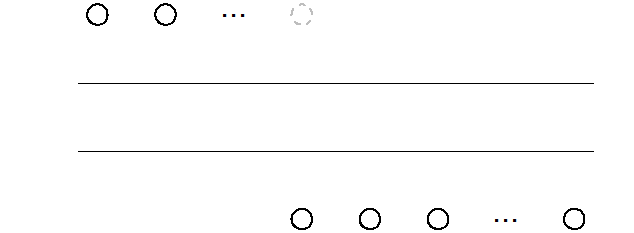}
        \subcaption{}
        \label{subfig:3}
    \end{subfigure}
    \hfill
    \begin{subfigure}[b]{0.49\linewidth}
        \centering
        \includegraphics[width=\linewidth]{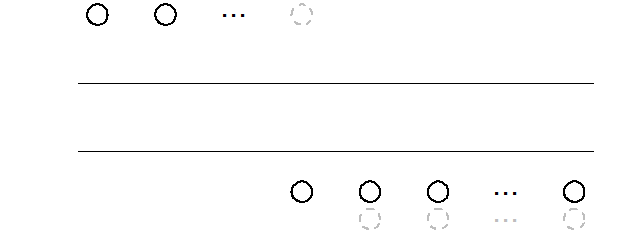}
        \subcaption{}
        \label{subfig:4}
    \end{subfigure}
    
    \caption{
    \textcolor{blue}{
    Sketches of four joint parameters.
    In each sketch, vertically, the region above the upper line denotes $\Theta^1$, and the region below the lower line denotes $\Theta^0$;
    horizontally, each dot corresponds to a local parameter.
    Specifically, (a) represents the true $\btheta\in\cS$, where the local parameters on the left are equal to $\theta^1\in\Theta^1$ and the rest on the right equal to  $\theta^0\in\Theta^0$.
    (b) represents the element in $\Theta^K$ that moves one $\theta^1$ to the boundary of $\Theta^0$, which achieves $\tilde I^1(\btheta)$.
    (c) represents the element in $\cS$ that moves one $\theta^1$ to $\theta^0$, which is easy to be misjudged as the one that achieves $I^1(\btheta)$.
    (d) represents the true element in $\cS$ that achieves $I^1(\btheta)$, which moves one $\theta^1$ and all $\theta^0$'s to a common point in $\Theta^0$ that is closer to the boundary than $\theta^0$ is.}
    }
    \label{Fig: sketches}
\end{figure}

If we denote the joint parameter in Fig.~\ref{Fig: sketches}.(a) as $\btheta=(\theta^1,\ldots,\theta^1,\theta^0,\cdots,\theta^0)\in\cS$, then the KL divergences between $\btheta$ and the joint parameters in Fig.~\ref{Fig: sketches}.(b)(c)(d) are $\tilde I^1(\btheta)=I(\theta^1,\Theta^0)$,
$I(\theta^1,\theta^0)$, and 
\begin{equation*}
    I^1(\btheta)=\inf_{\theta\in\Theta^0} \{I(\theta^1,\theta) + (K-|A(\btheta)|) I(\theta^0,\theta)\},
\end{equation*}
respectively.
It is clear that
$\tilde I^1(\btheta)$ does not satisfy the structure of $\cS$ unless $A(\btheta)=[K]$ or $\theta^0$ is on the boundary of $\Theta^0$, 
and that
$I(\theta^1,\theta^0)\geq I^1(\btheta)$ since $\theta=\theta^0$ is feasible for the minimization problem defining $I^1(\btheta)$.

When $\theta$ represents the mean parameter of the Gaussian distribution with unit variance, so that $I(\theta,\theta')=(\theta-\theta')^2/2$, 
$I^1(\btheta)$ can be specified as
\begin{equation*}
\begin{aligned}
     & \; \inf_{\theta\in\Theta^0} \left\{ \frac{1}{2}(\theta^1-\theta)^2 + (K-|A(\btheta)|)\frac{1}{2}(\theta^0-\theta)^2 \right\} \\
    = & \; \inf_{\theta^0+\delta\in\Theta^0}
    \left\{ \frac{1}{2}(\theta^1-\theta^0-\delta)^2 + (K-|A(\btheta)|)\frac{1}{2}\delta^2 \right\},
\end{aligned}
\end{equation*}
so
$I(\theta^1,\theta^0)-I^1(\btheta)$
is equal to 
\begin{equation*}
\begin{aligned}
    & \; \sup_{\theta^0+\delta\in\Theta^0}\left\{(\theta^1-\theta^0)\delta-(K-|A(\btheta)|+1)\frac{\delta^2}{2}\right\} \geq 0,
\end{aligned}
\end{equation*}
with equality if and only if $\theta^0$ is on the boundary of $\Theta^0$. 
}

\section{The proposed test} \label{section: the proposed test}
\subsection{Notations}
For any $k\in[K]$ and $n\geq 1$, denote by 
$L_k(n;\theta)\equiv \prod_{t=1}^n f_\theta(X_k(t))$ 
the likelihood function for the local parameter in stream $k$ based on its first $n$ data,
by 
\begin{equation*}
    \hat\theta_k(n) \equiv \argmax\{L_k(n;\theta): \theta\in\Theta\}
\end{equation*}
the maximum likelihood estimator (MLE) in $\Theta$,
{\color{blue}
by
\begin{equation*}
\begin{aligned}
    L_k^i(n) & \equiv \sup\{ L_k(n;\theta): \theta\in\Theta^i \} \text{ for } i\in\{0,1\} \\
\end{aligned}
\end{equation*}
the maximum likelihoods in $\Theta^i$,
and by
\begin{equation*}
    \hat A(n) \equiv \{k\in[K]: L_k^1(n) \geq L_k^0(n)\}
\end{equation*}
an estimator of the subset of signals.
}

We define the \emph{adaptive joint likelihood} as $\hat\bL(0)\equiv 1$ and,
for any $n\geq 1$, 
\begin{equation*} 
    \hat\bL(n) \equiv \hat\bL(n-1) \cdot \prod_{k\in[K]} f_{\hat\theta_k(n-1)}(X_k(n)),
\end{equation*}
where $\hat\bL(n)/\hat\bL(n-1)$ is the joint density of the data at time $n$, $\{X_k(n),\,k\in[K]\}$, under the parameters $\{\hat\theta_k(n-1),\,k\in[K]\}$ that are estimated based on the data up to time $n-1$, and $\{\hat\theta_k(0),\,k\in[K]\}$ are arbitrary initializations.
Besides, denote by $\bL(n;\btheta) \equiv \prod_{k\in[K]} L_k(n;\theta_k)$
the likelihood function for the joint parameter,
and define
\begin{equation} \label{genl}
\begin{aligned}
    \bL^i(n) \equiv \sup\left\{ \bL(n;\btheta'):\btheta'\in\Alt^i(\hat A(n)) \right\}
\end{aligned}
\end{equation}
for $i\in\{0,1\}$
as the \emph{maximum joint likelihoods} with respect to the subset of $\cS$ under which type-II (resp. I) errors would be made if the test decided on $\hat A(n)$. 

{\color{blue}
\begin{remark}
    Similarly, we denote $\tilde\bL^i(n)$ as $\bL^i(n)$ with $\btheta'\in\tilde\Alt^i(\hat A(n))$.
    The key that makes the proposed test suitable for the shared-parameter structure is the definition of these two maximum joint likelihoods in \eqref{genl}, or more precisely, the recognition of these two alternative subsets.
    This step is trivial in the unstructured setting, but becomes essential when the joint parameter has certain structure.
    Moreover, the same idea is readily extendable to other structures beyond the one considered in this work.
\end{remark}}

\subsection{Description}
Suppressing dependence on the two thresholds $a,b>0$, the proposed test $(\hat T,\hat D)$ is defined as:
\begin{equation} \label{the proposed test} 
\begin{aligned}
    \hat T & \equiv \inf\left\{n\geq 1: \frac{\hat\bL(n)}{\bL^0(n)} \geq b  \text{ and } \frac{\hat\bL(n)}{\bL^1(n)} \geq a \right\},
\end{aligned}
\end{equation}
and $\hat D \equiv \hat A(\hat T)$,
i.e., 
we sample until both $\hat\bL(n)/\bL^0(n)\geq b$ and $\hat\bL(n)/\bL^1(n)\geq a$, at which time we declare $\hat A(\hat T)$ as the subset of signals.
Intuitively,
$\hat\bL(n)/\bL^0(n)$ (resp. $\hat\bL(n)/\bL^1(n)$) represents evidence against all parameters under which we will make type-II (resp. I) errors if we decide on $\hat A(n)$. 
Therefore, when $\hat\bL(n)/\bL^0(n)$ (resp. $\hat\bL(n)/\bL^1(n)$) is large enough, deciding on $\hat A(n)$ is safe against type-II (resp. I) errors.

{\color{blue}
\begin{remark}
    Similarly, the state-of-the-art test for the unstructured setting, the Intersection rule in \cite{Song_prior, PaperIII}, takes the following form:
    \begin{equation} \label{the SOTA test} 
    \begin{aligned}
        \tilde T & \equiv \inf\left\{n\geq 1: \frac{\hat\bL(n)}{\tilde\bL^0(n)} \geq b  \text{ and } \frac{\hat\bL(n)}{\tilde\bL^1(n)} \geq a \right\}, 
    \end{aligned}
    \end{equation}
    and $\tilde D \equiv \hat A(\tilde T)$.
    Properties of this test similar to Theorem \ref{theorem, error control} and \eqref{theorem, UB} below are presented in Section \ref{supplement: section, properties of the SOTA test} of the supplement.
\end{remark}}

\subsection{Error control and choice of thresholds}
In this subsection, we show how to select the thresholds of the proposed test, $a$ and $b$, to control the two types of familywise error rates below desired levels.
\begin{theorem} \label{theorem, error control}
    Suppose that, for every $\btheta\in\cS$ and $i\in\{0,1\}$,
    \begin{equation} \label{assumption for error control}
        \Pro_\btheta\left( \lim_{n\to\infty}
        \frac{\hat\bL(n)}{\bL^i(n)} 
        = \infty \right) = 1.
    \end{equation}
    Then, for every $\btheta\in\cS$ and $a,b>0$, we have $\Pro_\btheta(\hat T<\infty)=1$ and 
    \begin{equation*} 
        \Pro_\btheta(\hat D\backslash A(\btheta)\neq\emptyset) \leq 1/a, \quad 
        \Pro_\btheta(A(\btheta)\backslash \hat D\neq\emptyset) \leq 1/b.
    \end{equation*}
    Thus, $(\hat T,\hat D)\in\Delta(\alpha,\beta)$ if we select $a=1/\alpha$, $b=1/\beta$.
\end{theorem}
    

Condition \eqref{assumption for error control} is very mild, which requires that the evidence in favor of the truth and against all wrong ones accumulates without bound as the sample size increases.

\subsection{Asymptotic upper bound and asymptotic optimality}

In this subsection, we establish an asymptotic upper bound on the expected sample size of the proposed test as $a,b\to\infty$, which, combined with the asymptotic lower bound in Theorem~\ref{theorem, LB} and the error control in Theorem~\ref{theorem, error control}, implies the asymptotic optimality.
\begin{theorem} \label{theorem, UB}
    Suppose that, for every $\btheta\in\cS$, $\epsilon>0$ and $i\in\{0,1\}$,
    \begin{equation} \label{assumption for AUB}
    \begin{aligned}
        & \sum_{n=1}^\infty \Pro_\btheta\left( \frac{1}{n} \log \frac{\hat\bL(n)}{\bL^i(n)} - I^i(\btheta) \leq \epsilon \right) < \infty.
    \end{aligned}
    \end{equation}
    Then, for every $\btheta\in\cS$, as $a,b\to\infty$, we have
    \begin{equation*} \label{asy UB}
        \Exp_\btheta[\hat T] \lesssim \max\left\{ \frac{\log b}{I^0(\btheta)}, \frac{\log a}{I^1(\btheta)} \right\}.
    \end{equation*}
\end{theorem}
\begin{remark} \label{remark for exponential family}
    Condition \eqref{assumption for AUB} strengthens condition \eqref{assumption for error control}. 
    As shown in the supplement, condition \eqref{assumption for AUB} is satisfied when $\{f_\theta,\,\theta\in\Theta\}$ belongs to an exponential family and the parameter spaces $\Theta^0,\Theta^1$ are compact.
    The exponential family distributions include most of commonly used distributions such as Gaussian, Uniform, Bernoulli, Poisson, etc.
    For more details, we refer to \citep[Appendix E]{Song_AoS}.
\end{remark}

\begin{corollary}
    If the thresholds of the proposed test, $a,b$, are selected so that $(\hat T,\hat D)\in\Delta(\alpha,\beta)$ for any $\alpha,\beta\in(0,1)$ and $a\sim|\log\alpha|, b\sim|\log\beta|$ as $\alpha,\beta\to 0$, e.g., as in Theorem \ref{theorem, error control}, then, for every $\btheta\in\cS$, as $\alpha,\beta\to 0$, we have
    \begin{equation} \label{asy approximation}                  
        \Exp_\btheta[\hat T] \sim \cL_\btheta(\alpha,\beta) \sim \max\left\{ \frac{|\log\beta|}{I^0(\btheta)}, \frac{|\log\alpha|}{I^1(\btheta)} \right\}.
    \end{equation}
\end{corollary}

\section{Numerical studies} \label{section: Numerical study}
{\color{blue}
In this section, we present numerical studies that compare the expected sample sizes (ESS) of the proposed test in \eqref{the proposed test} and the Intersection rule in \eqref{the SOTA test}, \emph{with the same levels of familywise error rates (Err)}.
Specifically, the ESS are estimated via plain Monte-Carlo, and the Err, as probabilities of rare events, via importance sampling.
Due to page limits, how importance sampling is conducted and extra numerical studies about deviation from the shared-parameter structure, are presented in Section \ref{supplement: section, numerical studies} of the supplement.


We assume that $f_\theta$ represents Gaussian distribution with mean $\theta$ and variance $1$, and $\Theta^0=(-\infty,-\delta]$, $\Theta^1=[\delta,+\infty)$ for some $\delta>0$.
We set $K=10$, $\delta=0.1$, and $\btheta$ such that $A(\btheta)=\{1,\ldots,5\}$, $\theta^0=-0.5$ and $\theta^1=0.5$, which give $I^0(\btheta)=I^1(\btheta)=5/12$ and $\tilde I^0(\btheta)=\tilde I^1(\btheta)=0.18$.
The ESS and their asymptotic approximations given by $-\operatorname{log}(\operatorname{Err})$ over $I^0(\btheta)$ or $\tilde I^0(\btheta)$ are shown in Fig.~\ref{Fig}.
We can see that the proposed test outperforms the Intersection rule uniformly, and the asymptotic approximations fit the actual performance well.
}

\begin{figure}
    \centering
    \includegraphics[width=0.85\linewidth]{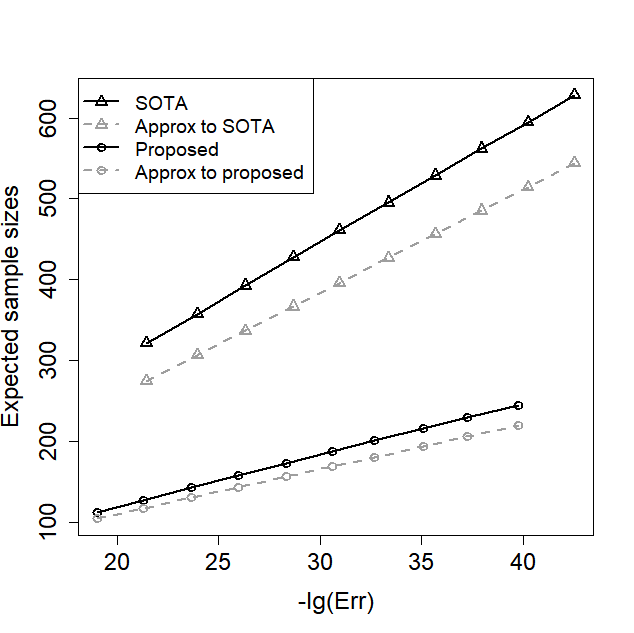}
    \caption{
    Expected sample sizes of the proposed test and the Intersection rule, against negative base-$10$ logarithm of their actual familywise error rates.
    The dashed gray lines represent their asymptotic approximations.
    }
    \label{Fig}
\end{figure}

\section{Future directions} \label{section: conclusion}
In this paper, we consider the full sampling setup where all streams are observed at every time instant. It is of interest to consider the active sampling setup where only part of the streams can be observed, which need to be chosen adaptively based on the previous data. 
This setup has been studied recently in various sequential settings 
\citep{Kobi_2020_composite, Qunzhi_2021TIT, Aris_IEEE, song2024change, chaudhuri2024round, Aris_TIT2025, JMVA2026, TSP2026},
and is closely related to the problem of multi-arm bandits pure exploration with fixed-confidence \citep{OAI, Overlapping2021, Ali2023, Ali2025}.
Another interesting direction is to consider more than two hypotheses and other structural assumptions on the joint distribution \citep{ICASSP2026}.

\appendices

\begin{center}
{\large\bf Supplementary Material to ``Signal Detection under Composite Hypotheses with Identical Distributions for Signals and for Noises"}
\end{center}

This supplement contains an extension to error metrics other than the familywise error rates (FWER), proofs that are omitted in the main text, and extra numerical studies. 

\section{Extension to other error metrics}
In this subsection, we follow the approach in \cite{Bartroff_2021_EquivErrorMetrics} and \citep[Section 7]{PaperIII} and extend the asymptotic optimality theory to error metrics other than the FWER. 
Specifically, we denote by 
\begin{equation*}
\begin{aligned}
    \operatorname{FWER}_\btheta^1(T,D) & \equiv \Pro_\btheta(D\backslash A(\btheta)\neq\emptyset), \\
    \operatorname{FWER}_\btheta^2(T,D) & \equiv \Pro_\btheta(A(\btheta)\backslash D\neq\emptyset),
\end{aligned}
\end{equation*}
the two types of FWER of the test $(T,D)$ when the true parameter value is $\btheta$, and by 
\begin{equation*}
    \operatorname{GEM}_\btheta^1(T,D), \quad 
    \operatorname{GEM}_\btheta^2(T,D),
\end{equation*}
the generic notations of a error metric (GEM) with two types.
For any $\alpha,\beta>0$, we denote by 
\begin{equation*}
\begin{gathered}
    \Delta^{\text{GEM}}(\alpha,\beta) \equiv \{
    (T,D)\in\Delta: \;
    \Pro_\btheta(T<\infty) = 1, \\
    \operatorname{GEM}_\btheta^1(T,D) \leq \alpha, \; 
    \operatorname{GEM}_\btheta^2(T,D) \leq \beta, \; \forall\; \btheta\in\cS
    \}
\end{gathered}
\end{equation*}
the subfamily of tests that terminate almost surely and control the two types of GEMs below $\alpha,\beta$ respectively, and by 
\begin{equation*}
    \cL_\btheta^{\text{GEM}}(\alpha,\beta) \equiv \inf\{ \Exp_\btheta[T]: (T,D)\in\Delta(\alpha,\beta) \}, \; \forall\;\btheta\in\cS
\end{equation*}
the minimum expected sample size among all such tests. 
Then, we have the following proposition:
\begin{proposition} \label{proposition, GEM}
    Suppose there exist constants $C_1,C_2>0$ so that, for every $i\in\{1,2\}$ and $\btheta\in\cS$,
    \begin{align}
        \operatorname{GEM}_\btheta^i(\hat T,\hat D) & \leq C_1 \cdot \operatorname{FWER}_\btheta^i(\hat T,\hat D), \label{GEM} \\
        \operatorname{GEM}_\btheta^i(T,D) & \geq C_2 \cdot \operatorname{FWER}_\btheta^i(T,D) \text{ for all } (T,D)\in\Delta, \nonumber
    \end{align}
    where $(\hat T,\hat D)$ represents the proposed test. 
    Then, by selecting $a=C_1/\alpha$ and $b=C_2/\beta$ for any $\alpha,\beta>0$, we have $(\hat T,\hat D)\in\Delta^{\text{GEM}}(\alpha,\beta)$ for all $\alpha,\beta>0$ and
    \begin{equation*}
        \Exp_\btheta[\hat T] \sim \cL_\btheta^{\text{FWER}}(\alpha,\beta) \sim \cL_\btheta^{\text{GEM}}(\alpha,\beta)
    \end{equation*}
    as $\alpha,\beta\to 0$.
\end{proposition}
\begin{proof} [Proof of Proposition \ref{proposition, GEM}]
    Let $(\hat T,\hat D)(a,b)$ denote the proposed test with thresholds $a,b$. 
    Then, Theorem \ref{theorem, error control} of the main text and the first line of condition \eqref{GEM} imply 
    \begin{equation*}
        (\hat T,\hat D)(C_1/\alpha,C_1/\beta)\in \Delta^{\text{FWER}}(\alpha/C_1,\beta/C_1) \subseteq \Delta^{\text{GEM}}(\alpha,\beta).
    \end{equation*}
    
    Besides, the second line of condition \eqref{GEM} implies $\Delta^{\text{GEM}}(\alpha,\beta)\subseteq\Delta^{\text{FWER}}(\alpha/C_2,\beta/C_2)$, so
    \begin{equation*}
    \begin{aligned}
        \cL^{\text{GEM}}(\alpha,\beta) & \geq \cL_\btheta^{\text{FWER}}(\alpha/C_2,\beta/C_2) \sim \cL_\btheta^{\text{FWER}}(\alpha,\beta) \\
        & \sim \Exp_\btheta[\hat T(C_1/\alpha,C_1/\beta)] \geq \cL^{\text{GEM}}(\alpha,\beta),
    \end{aligned}
    \end{equation*}
    where the two $\sim$'s are based on Theorem \ref{theorem, error control} and the specific expressions of $\cL^{\text{FWER}}(\cdot,\cdot)$, and the second $\geq$ is because $(\hat T,\hat D)(C_1/\alpha,C_1/\beta)\in\Delta^{\text{GEM}}(\alpha,\beta)$.
\end{proof}
This proposition implies that, as long as GEM are upper (resp. lower) bounded by constant multiples of FWER for the proposed test (resp. for all tests),
we can select the thresholds of the proposed test in a very straightforward way to control GEM and achieve the same asymptotic optimality property with respect to GEM. 

Here are two examples for GEM:
Per-stream error rates (PSER):
\begin{equation*}
\begin{aligned}
    \operatorname{PSER}_\btheta^1(T,D) & \equiv \Exp_\btheta[|D\backslash A(\btheta)|]/K, \\
    \operatorname{PSER}_\btheta^2(T,D) & \equiv \Exp_\btheta[|A(\btheta)\backslash D|]/K,
\end{aligned}
\end{equation*}
and false discovery rates (FDR):
\begin{equation*}
\begin{aligned}
    \operatorname{FDR}_\btheta^1(T,D) & \equiv \Exp_\btheta[|D\backslash A(\btheta)|/|D|], \\
    \operatorname{FDR}_\btheta^2(T,D) & \equiv \Exp_\btheta[|A(\btheta)\backslash D|/(K-|D|)].
\end{aligned}
\end{equation*}
Since
\begin{equation*}
\begin{aligned}
    \frac{1}{K}1\{D\backslash A(\btheta)\neq\emptyset\} & \leq \frac{|D\backslash A(\btheta)|}{K} \\
    & \leq \frac{|D\backslash A(\btheta)|}{|D|} \leq 1\{D\backslash A(\btheta)\neq\emptyset\},   
\end{aligned}
\end{equation*}
we have 
\begin{equation*}
\begin{aligned}
    \frac{1}{K} \operatorname{FWER}_\btheta^1(T,D) & \leq \operatorname{PSER}_\btheta^1(T,D) \\
    & \leq \operatorname{FDR}_\btheta^1(T,D) \leq \operatorname{FWER}_\btheta^1(T,D).    
\end{aligned}
\end{equation*}
Similarly for the other type. 
Therefore, this proposition applies to PSER and FDR with $C_1=1$ and $C_2=1/K$.

\section{Proofs}
\begin{proof} [Proof of Theorem \ref{theorem, LB}]
    Fix arbitrary $\btheta\in\cS$, $\alpha,\beta\in(0,1/2)$ that $\alpha+\beta<1/2$, and $(T,D)\in\Delta(\alpha,\beta)$.
    Also fix arbitrary $\btheta'\in\cS$ that $\btheta\neq\btheta'$.
    By Wald's identity, 
    \begin{equation*}
    \begin{aligned}
        & \; \Exp_\btheta\left[ \sum_{t=1}^T \log \frac{f_{\btheta}(X_1(t),\ldots,X_K(t))}{f_{\btheta'}(X_1(t),\ldots,X_K(t))} \right] \\
        = & \; \Exp_\btheta[T] \, \Exp_\btheta\left[ \log \frac{f_{\btheta}(X_1(1),\ldots,X_K(1))}{f_{\btheta'}(X_1(1),\ldots,X_K(1))} \right]
        = \Exp_\btheta[T] \, I(\btheta,\btheta').
    \end{aligned}
    \end{equation*}
    Meanwhile, by the information-theoretical inequality (see, e.g., Lemma 3.2.1 of \cite{Tartakovsky_Book}) and the fact that $D$ is $\cF(T)$-measurable, the left-hand-side of the above equation is lower bounded by 
    \begin{equation*}
    \begin{aligned}
        & \; \varphi\left( \Pro_\btheta(D\neq A(\btheta)), \Pro_{\btheta'}(D=A(\btheta)) \right) \\
        \geq & \; \varphi(\alpha+\beta, \Pro_{\btheta'}(D=A(\btheta))), 
    \end{aligned}
    \end{equation*}
    where we also used the fact that $\varphi(\cdot,\cdot)$ is decreasing in both arguments and that $D\neq A(\btheta)$ makes at least one error when the joint parameter is $\btheta$.
    
    Now, if $\btheta'\in\cS$ so that $A(\btheta')\backslash A(\btheta)\neq\emptyset$, we also have 
    $\Pro_{\btheta'}(D=A(\btheta))\leq\beta$ since $D=A(\btheta)$ makes at least one type-II error when the joint parameter is $\btheta'$. 
    Therefore, 
    $$\varphi(\alpha+\beta, \Pro_{\btheta'}(D=A(\btheta))) \geq \varphi(\alpha+\beta,\beta).$$
    Combining with the above results,
    it follows that 
    $$\Exp_\btheta[T] \geq \varphi(\alpha+\beta,\beta)/I(\btheta,\btheta').$$
    Since this holds for all such $\btheta'\in\cS$, we obtain the first term in the desired lower bound.
    The second term can be obtained analogously and thus is omitted.
\end{proof}

\begin{proof} [Proof of Theorem \ref{theorem, error control}]
    Fix arbitrary $\btheta\in\cS$ and $a,b>0$.
    Condition \eqref{assumption for error control} implies that 
    \begin{equation*}
        \Pro_\btheta\left( \lim_{n\to\infty}\frac{\hat\bL(n)}{\bL^i(n)} = \infty \text{ for } i\in\{0,1\} \right) = 1.
    \end{equation*}
    Thus, $\Pro_\btheta(\hat T<\infty)=1$.
    
    Next, we only show the upper bound on the familywise error rate of type-I, as that of type-II can be shown similarly.
    Indeed, 
    \begin{equation*}
    \begin{aligned}
        \{\hat D\backslash A(\btheta)\neq\emptyset\} 
        & \subseteq \left\{\frac{\hat\bL(T)}{\bL^1(T)}\geq a, \; \bL^1(T) \geq \bL(T;\btheta)\right\} \\
        & \subseteq \left\{ \exists\;n\geq 1, \, \frac{\hat\bL(n)}{\bL(n;\btheta)} \geq a \right\}.
    \end{aligned}
    \end{equation*}
    From Lemma \ref{lemma, martingale} of this supplement, we know that $\{\hat\bL(n)/\bL(n;\btheta),\;n\geq 1\}$ is a non-negative, mean-one martingale under $\Pro_\btheta$. Then the desired result follows from Ville's inequality.
\end{proof}

\begin{proof} [Proof of Theorem \ref{theorem, UB}]
    With condition \eqref{assumption for AUB}, applying Lemma \ref{lemma, for AUB} of this supplement
    to the stopping time in \eqref{the proposed test}, the asymptotic upper bound follows. 
\end{proof}

\begin{proof} [Proof of Remark \ref{remark for exponential family}]
    We only prove \eqref{assumption for AUB} for $i=0$ as that for $i=1$ is analogous.
    Indeed, each summand is equal to 
    \begin{equation*} 
    \begin{aligned} 
        & \; \Pro_\btheta\left(\frac{1}{n} \log \frac{\hat\bL(n)}{\bL^0(n)} - I^0(\btheta) \leq \epsilon, \; \hat A(n)=A(\btheta)\right) + \\
        & \; \Pro_\btheta\left(\frac{1}{n} \log \frac{\hat\bL(n)}{\bL^0(n)} - I^0(\btheta) \leq \epsilon, \; \hat A(n)\neq A(\btheta)\right) \\
        \leq & \; \Pro_\btheta\left(\frac{1}{n} \log \frac{\hat\bL(n)}{\bL^0(n;A(\btheta))} - I^0(\btheta) \leq \epsilon\right) + \\
        & \; \Pro_\btheta\left(\hat A(n)\neq A(\btheta)\right),
    \end{aligned}
    \end{equation*}    
    where, for simplicity, we denote 
    \begin{equation*}
        \bL^0(n;A(\btheta)) \equiv \sup\left\{ \bL(n;\btheta'): \btheta'\in\cS, \; A(\btheta')\backslash A(\btheta) \neq\emptyset \right\}.
    \end{equation*}
    The summation of the first term is finite based on \citep[Lemma E.1]{Song_AoS} since $\{\btheta'\in\cS: A(\btheta')=A(\btheta)\}$ and $\{\btheta'\in\cS: A(\btheta')\backslash A(\btheta) \neq\emptyset\}$ are two disjoint, compact sets,   
    and the summation of the second term is also finite following a similar reasoning as ``Step 2" in the proof of \citep[Theorem 2]{Kobi_2020_composite}.
\end{proof}

\section{Supporting lemmas}
\begin{lemma} \label{lemma, martingale}
    For every $\btheta\in\cS$, 
    \begin{equation*}
        \left\{ \frac{\hat \bL(n)}{\bL(n;\btheta)},\;n\geq 1 \right\}
    \end{equation*}
    is a non-negative, mean-one martingale with respect to filtration $\cF$ under measure $\Pro_\btheta$.
\end{lemma}
\begin{proof} [Proof of Lemma \ref{lemma, martingale}]
    By definition, 
    \begin{equation*}
        \frac{\hat\bL(n)}{\bL(n;\btheta)} = \frac{\hat\bL(n-1)}{\bL(n-1;\btheta)}\frac{\prod_{k\in[K]} f_{\hat\theta_k(n-1)}(X_k(n))}{\prod_{k\in[K]} f_{\theta_k}(X_k(n))}
    \end{equation*}
    for $n\geq 1$,
    with $\hat\bL(0)/\bL(0;\btheta)=1$, $\{\hat\theta_k(0),\,k\in[K]\}$ deterministic and $\{\hat\theta_k(n-1),\,k\in[K]\}$ $\cF(n-1)$-measurable.
    It is easy to see that this stochastic process is non-negative and adapted to filtration $\cF$.
    To see it is a martingale under $\Pro_\btheta$, note that, since $\{\hat\theta_k(n-1),\,k\in[K]\}$ are $\cF(n-1)$-measurable,
    \begin{equation*}
    \begin{aligned}
        & \; \Exp_\btheta\left[\left. \frac{\prod_{k\in[K]} f_{\hat\theta_k(n-1)}(X_k(n))}{\prod_{k\in[K]} f_{\theta_k}(X_k(n))} \right|\cF(n-1)\right] \\
        = & \; \int \frac{\prod_{k\in[K]} f_{\hat\theta_k(n-1)}(z_k)}{\prod_{k\in[K]} f_{\theta_k}(z_k)} \, \prod_{k\in[K]} f_{\theta_k}(z_k) \, (d\nu)^K = 1.
    \end{aligned}
    \end{equation*}
    To see it has mean one, note that, since $\{\hat\theta_k(0),\,k\in[K]\}$ are deterministic, 
    \begin{equation*}
    \begin{aligned}
        & \; \Exp_\btheta\left[ \frac{\hat\bL(1)}{\hat\bL(1;\btheta)} \right] \\
        = & \; \int \frac{\prod_{k\in[K]} f_{\hat\theta_k(0)}(z_k)}{\prod_{k\in[K]} f_{\theta_k}(z_k)} \, \prod_{k\in[K]} f_{\theta_k}(z_k) \, (d\nu)^K = 1.
    \end{aligned}
    \end{equation*}
\end{proof}

\begin{lemma} \label{lemma, for AUB}
    Let $\{\xi_k(n),\,n\geq 1\}$, $k\in[K]$ be $K\geq 1$ stochastic processes on some probability space with measure $\Pro$.
    For any $\bb=(b_1,\ldots,b_K)\in(0,\infty)^K$, define stopping time
    \begin{equation*}
        T(\bb) \equiv \inf\{n\geq 1: \xi_k(n)\geq b_k \text{ for all } k\in[K]\}.
    \end{equation*}
    If, for any $k\in[K]$, there exists $\mu_k>0$ so that for any $\epsilon>0$, 
    \begin{equation*}
        \sum_{n=1}^\infty \Pro\left( \frac{1}{n} \xi_k(n) - \mu_k \leq \epsilon \right) < \infty,
    \end{equation*}
    then
    \begin{equation*}
        \Exp[T(\bb)] \lesssim \max_{k\in[K]} \left\{\frac{b_k}{\mu_k}\right\} \text{ as } \min_{k\in[K]} b_k\to\infty.
    \end{equation*}
\end{lemma}
\begin{proof} [Proof of Lemma \ref{lemma, for AUB}]
    See \citep[Lemma F.2]{Song_AoS}.
\end{proof}

\section{Properties of the Intersection rule} \label{supplement: section, properties of the SOTA test}
In this section, we summarize the theoretical properties of the Intersection rule in \cite{Song_prior, PaperIII} that was designed for the problem of sequential multiple testing in the general, unstructured setting. This test has been studied in \cite{Song_prior, PaperIII} and the proofs of the following results can be found there.
This test is repeated in \eqref{supplement, the SOTA test} for convenience. 
\begin{theorem} \label{supplement: theorem, error control}
    Suppose that, for every $\btheta\in\Theta^K$ and $i\in\{0,1\}$,
    \begin{equation} \label{supplement: assumption for error control}
        \Pro_\btheta\left( \lim_{n\to\infty}
        \frac{\hat\bL(n)}{\tilde\bL^i(n)} 
        = \infty \right) = 1.
    \end{equation}
    Then, for every $\btheta\in\Theta^K$ and $a,b>0$, we have $\Pro_\btheta(\tilde T<\infty)=1$ and 
    \begin{equation*} 
        \Pro_\btheta(\tilde D\backslash A(\btheta)\neq\emptyset) \leq 1/a, \quad 
        \Pro_\btheta(A(\btheta)\backslash \tilde D\neq\emptyset) \leq 1/b.
    \end{equation*}
    Thus, $(\tilde T,\tilde D)\in\tilde\Delta(\alpha,\beta) \subseteq \Delta(\alpha,\beta)$ if we select $a=1/\alpha$, $b=1/\beta$.
\end{theorem}

\begin{theorem} \label{supplement: theorem, UB}
    Suppose that, for every $\btheta\in\Theta^K$, $\epsilon>0$ and $i\in\{0,1\}$,
    \begin{equation} \label{supplement: assumption for AUB}
    \begin{aligned}
        & \sum_{n=1}^\infty \Pro_\btheta\left( \frac{1}{n} \log \frac{\hat\bL(n)}{\tilde\bL^i(n)} - \tilde I^i(\btheta) \leq \epsilon \right) < \infty.
    \end{aligned}
    \end{equation}
    Then, for every $\btheta\in\Theta^K$, as $a,b\to\infty$, we have
    \begin{equation} \label{supplement: asy UB}
        \Exp_\btheta[\tilde T] \sim \max\left\{ \frac{\log b}{\tilde I^0(\btheta)}, \frac{\log a}{\tilde I^1(\btheta)} \right\}.
    \end{equation}
\end{theorem}
\begin{remark}
    Compared with Theorem \ref{theorem, error control} and \ref{theorem, UB} in the main text, both the assumptions and the conclusions are changed from ``for all $\btheta\in\cS$" to ``for all $\btheta\in\Theta^K$".  
    However, the asymptotic approximation for Intersection rule in \eqref{supplement: asy UB} is always greater than or equal to the one for the proposed test in \eqref{asy approximation}, since $\tilde I^0(\btheta)\leq I^0(\btheta)$ and $\tilde I^1(\btheta)\leq I^1(\btheta)$.
\end{remark}

\section{Numerical studies} \label{supplement: section, numerical studies}
In this section, we first discuss how we use importance sampling to estimate the actual familywise error rates of the two tests, 
and then visualize how the two tests perform under model-misspecification. 

Before starting, we restate the setup and the definition of the two tests:
Suppose that $f_\theta$ represents Gaussian distribution with mean $\theta$ and variance $1$, and $\Theta^0=(-\infty,-\delta]$, $\Theta^1=[\delta,+\infty)$ for some $\delta>0$, 
i.e., we are testing whether the Gaussian mean is negative or positive with indifference zone $(-\delta,\delta)$.
We set $K=10$, $\delta=0.1$, and $\btheta$ such that $A(\btheta)=\{1,\ldots,5\}$, $\theta^0=-0.5$ and $\theta^1=0.5$, which give $I^0(\btheta)=I^1(\btheta)=5/12$ and $\tilde I^0(\btheta)=\tilde I^1(\btheta)=0.18$.
The proposed test is defined as 
\begin{equation} \label{supplement, the proposed test} 
\begin{aligned}
    \hat T & \equiv \inf\left\{n\geq 1: \frac{\hat\bL(n)}{\bL^0(n)} \geq b  \text{ and } \frac{\hat\bL(n)}{\bL^1(n)} \geq a \right\}, \\
    \hat D & \equiv \hat A(\hat T),
\end{aligned}
\end{equation}
and the Intersection rule in \cite{Song_prior, PaperIII} as 
\begin{equation} \label{supplement, the SOTA test} 
\begin{aligned}
    \tilde T & \equiv \inf\left\{n\geq 1: \frac{\hat\bL(n)}{\tilde\bL^0(n)} \geq b  \text{ and } \frac{\hat\bL(n)}{\tilde\bL^1(n)} \geq a \right\}, \\
    \tilde D & \equiv \hat A(\tilde T),
\end{aligned}
\end{equation}
where $\bL^i(n)$'s (resp. $\tilde\bL^i(n)$'s) represent the maximum joint likelihoods among wrong parameters with respect to $\hat A(n)$ in the structured space $\cS$ (resp. unstructured space $\Theta^K$).

\subsection{Importance sampling}
We focus on the estimation of the actual familywise type-II error rate, which is defined as
\begin{equation*}
    \Err \equiv \max_{\btheta\in\cS} \Pro_\btheta(A(\btheta)\backslash D\neq\emptyset).
\end{equation*}
Based on the setup of our numerical studies, it is easy to see that the maximum is attained when all streams are signals and the signal parameter is on the boundary, i.e.,
\begin{equation*}
    \Err = \Pro_{\btheta^*}(A(\btheta^*)\backslash D\neq\emptyset) \text{ where } \btheta^*\equiv(\delta_1,\ldots,\delta_1).
\end{equation*}

For generality, we fix arbitrary $\btheta\in\cS$ such that $A(\btheta)\neq\emptyset$ and estimate $\Pro_\btheta(A(\btheta)\backslash D\neq\emptyset)$.
Since this probability is very small (exponential decaying in threshold), computing it based on plain Monte-Carlo is very inefficient. 
Indeed, in order to make the relative error of the estimate 
\begin{equation*}
    \frac{sd(\hat p)}{p} = \left. \sqrt{\frac{p(1-p)}{n}} \right/ p \approx \frac{1}{\sqrt{np}}
\end{equation*}
as small as $5\%$, the number of plain Monte-Carlo rounds needs to be as large as $400/p$, which is unacceptable when $p$ is very small (in our numerical studies, we estimate $p$ as small as $10^{-50}$).

In this case, importance-sampling is the right tool (see, e.g., \cite{Siegmumd_IS, Bucklew_Book, Song2025_IS}).
The basic idea is that, in order to estimate a very small $\Pro(\Gamma)$, we find another distribution $\Qro$ such that (i) $\Qro(\Gamma)$ is large and (ii) $\Qro$ is close to $\Pro$, where the former guarantees that we are able to collect enough effective observations and the latter controls the variance of the importance sampling weights.
Then, based on the following Wald's likelihood ratio identity:
\begin{equation*}
    \Pro(\Gamma) = \Exp_\Pro[1_\Gamma] = \Exp_\Qro\left[ \left(\frac{d\Qro}{d\Pro}\right)^{-1} 1_\Gamma \right],
\end{equation*}
we estimate $\Pro(\Gamma)$ by simulating $(d\Qro/d\Pro)^{-1}1_\Gamma$ under $\Qro$ for many rounds and taking the average.

Specifically in our problem, in order to estimate $\Pro_\btheta(A(\btheta)\backslash D\neq\emptyset)$, we adopt the importance sampling distribution that uniform-randomly picks a signal and changes its local parameter from $\theta^1$ to $\theta^0$, i.e.,
\begin{equation*}
    \frac{1}{|A(\btheta)|} \sum_{k\in A(\btheta)} \Pro_{\btheta(k)},
\end{equation*}
where $\btheta(k)$ represents $\btheta$ with the $\theta^1$ at the $k_{th}$ position replaced by $\theta^0$. 
If $A(\btheta)=[K]$ so that there is no $\theta^0$, e.g., $\btheta^*$, we change $\theta^1$ to the boundary value $-\delta$.
Based on this importance sampling technique,
all relative errors in our numerical studies are below $5\%$ with $10^4$ simulation rounds. 

\subsection{Model misspecification}
In Fig. \ref{Fig_misspe}, we set thresholds $\log a=\log b=20$ and plot the expected sample sizes of the proposed test and the Intersection rule when $\theta_1$, the local parameter of the first stream, takes value in $\{0.1,0.2,\ldots,1\}$, where the value consistent with the shared-parameter structure is $0.5$. 
First note that larger $\theta_1$ means easier problem, so the expected sample size of both tests should not increase as $\theta_1$ increases, which is corroborated by the result.
Besides, we can see that the ESS of the Intersection rule decreases as $\theta_1$ increases to $0.5$ and basically keeps constant afterwards. 
This is because that the sample size of the Intersection rule is mainly determined by the most difficult stream, which is stream $1$ when $\theta_1<0.5$ and other streams otherwise. 
On the contrary, the ESS of the proposed test keeps decreasing, even when $\theta_1>0.5$, is never too large, and is always below that of the Intersection rule, demonstrating satisfactory robustness against deviation from the shared-parameter structure. 
\begin{figure}
    \centering
    \includegraphics[width=0.48\linewidth]{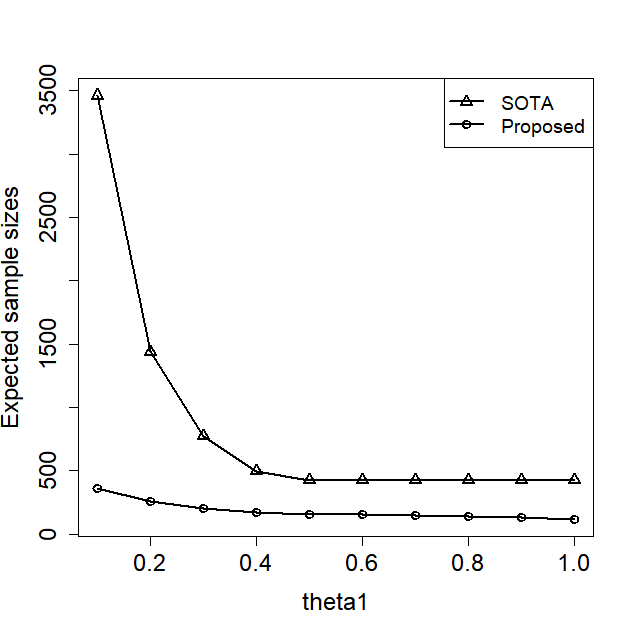}
    \includegraphics[width=0.48\linewidth]{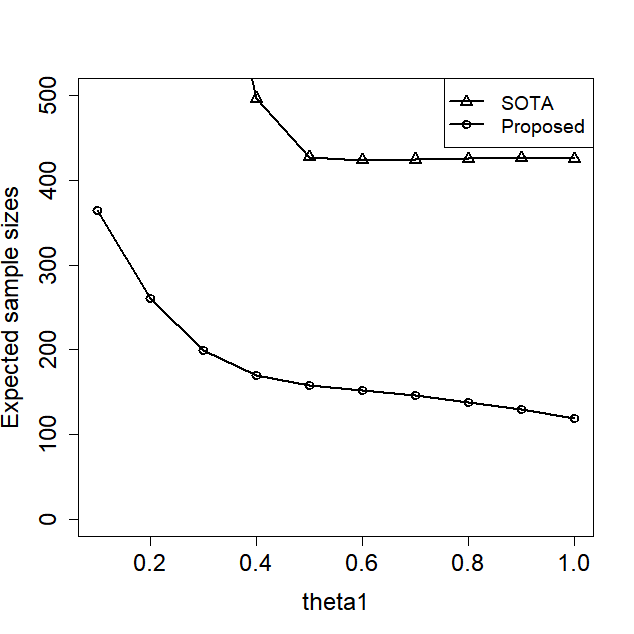}
    \caption{Expected sample sizes of the proposed test and the Intersection rule with thresholds $\log a = \log b = 20$, when $\theta_1$ ranges in $\{0.1,0.2,\ldots,1\}$. The right subfigure is the left subfigure when the y-axis is limited to $[0, 500]$.}
    \label{Fig_misspe}
\end{figure}

For the effects of stream number $K$ and class balance $|A(\btheta)|$ versus $K-|A(\btheta)|$, they are similar to those observed in other sequential multiple testing problems, so we refer to \citep[Supplement Section 1]{PaperIII} for details.


\bibliographystyle{chicago}
\bibliography{main}

\end{document}